\documentclass[a4paper,12pt,english,oneside,final]{article}
\usepackage[utf8]{inputenc} 
\usepackage[T1]{fontenc} 
\usepackage[english]{babel}
\usepackage{setspace}

\usepackage{lmodern} 

\usepackage{stackengine}
\usepackage[natbibapa]{apacite} 

\usepackage{amsthm,amssymb,mathrsfs,amsmath,amsfonts}
\usepackage{mathtools}

\usepackage{lineno} 

\usepackage{thmtools,thm-restate} 
\usepackage{euscript} 
\usepackage{verbatim}
\usepackage{subcaption}
\usepackage{enumitem}
\usepackage{microtype}
\usepackage{dsfont}
\usepackage{csquotes}
\usepackage{breqn}

\usepackage{tabulary}
\usepackage{booktabs, siunitx}
\usepackage[svgnames,table]{xcolor}
\usepackage[tableposition=below]{caption} 

\usepackage[bottom]{footmisc}
\usepackage{graphicx}

\definecolor{coolblack}{rgb}{0.0, 0.18, 0.39}

\usepackage[colorlinks=true,linkcolor=coolblack,
urlcolor=coolblack,
citecolor=coolblack,pdftex]{hyperref} 
\usepackage[all]{hypcap}

\usepackage{geometry}
\makeatletter

\newtheorem{proposition}{Proposition}
\newtheorem{corollary}{Corollary}

\usepackage[shadow,colorinlistoftodos]{todonotes}

\newcommand{\ubar}[1]{\text{\b{$#1$}}}

\makeatother

\usepackage{babel}

\DeclareMathOperator*{\argmax}{arg\,max}

\begin{document}
	
	\pagestyle{plain}
	
	\author{Federico Vaccari\thanks{Department of Economics, University of Bergamo, Italy. E-mail: \href{mailto:vaccari.econ@gmail.com}{\sf vaccari.econ@gmail.com}. I am grateful to Yair Antler, Daniele Condorelli, Harry Di Pei, Edoardo Grillo, Navin Kartik, Wataru Kitano, Matthew Mitchell, Santiago Oliveros, Federico Quartieri, Patrick Rey, and anonymous referees for helpful comments and suggestions. I also thank participants at the 2022 YETI Workshop, 2023 Oligo Workshop, APSA's Formal Theory Virtual Workshop, University of Chieti-Pescara, and the 2023 Spring Workshop on Economic Theory at LUISS. All errors are mine.}}

	\title{Efficient Communication in Organizations}

	\date{ }
	
	\maketitle

	\begin{abstract}
				\noindent This paper studies the organization of communication between biased senders and a receiver. Senders can misreport their private information at a cost. Efficiency is achieved by clearing information asymmetries without incurring costs. Results show that only one communication protocol is efficient, robust to collusion, and free from unnecessary complexities. This protocol has a simple, adversarial, and public structure. It always induces efficient equilibria, for which a closed-form characterization is provided. The findings are relevant for the design of organizations that seek to improve decision-making while limiting wasteful influence activities.
	\end{abstract}

	
	
	
	
	\clearpage

\section{Introduction}

Communication is an essential part of organizations. Decision-makers (owners, top management) often rely on better-informed parties within the same firm (lower-level managers, division heads) to provide relevant information. An agency problem arises when the informed parties' objectives are not aligned with the decision-makers'. In these cases, there are two potential sources of inefficiency. First, informed agents may dissipate considerable resources in influence activities. Second, decision-makers may take wrong decisions due to being poorly informed or swayed. As a result, organizations must structure communication to minimize wasteful influence activities while maximizing decision-makers' accuracy. This paper is concerned with communication protocols that achieve this goal.

To analyze communication protocols, I study a costly signaling game between an uninformed receiver and one or more informed senders. The receiver must select one of two actions. Senders know which action is better for the receiver. Before making a decision, the receiver interacts with the senders in a way pre-determined by a communication protocol. When communicating, senders can misreport information at a cost that is tied to the size of the lie. These ``misreporting costs'' represent the resources senders allocate to influence the receiver's decision. For example, a manager who shifts subordinates' labor to fabricate data will have fewer resources for the organization's production activities. As a result, the manager may underperform or incur penalties. From the organization's standpoint, these influence activities are unproductive and wasteful. 

The paper's central finding is that, among all possible arrangements, only one protocol has three desirable properties at the same time: it is efficient, robust to collusion, and minimal. This protocol, referred to as \emph{public advocacy}, consists of the sequential and public consultation of agents with opposed interests. It is a remarkably simple mechanism requiring neither commitment nor mediation. The results concern uniqueness as well as existence: a full equilibrium characterization shows that there always exists an efficient and robust equilibrium induced by the public advocacy protocol.

The notion of efficiency considered here is natural: it requires that all players obtain their complete-information payoff. An efficient protocol is critical in organizations because it solves asymmetric information problems without dissipating resources in unproductive signaling. Likewise, collusion is a common concern in organizations where members can discuss their intentions with each other, and it is relevant for protocols with multiple senders. Minimality requires protocols to have a simple structure, defined solely by the senders' biases and confidentiality. All else being equal, organizations typically prefer adopting simpler rather than complicated protocols.

Communication protocols specify the organization's mode of communication: how many senders to consult, their relative standing over decision-making, and the confidentiality of their recommendations. For example, a protocol may instruct the receiver to consult only one sender that favors a particular action (such as the head of a division that would benefit from a specific investment). Alternatively, it may specify to consult in a private setting two senders with aligned goals (e.g., members of the same department). In this last case, it is crucial to ensure that communication is not compromised by senders' collusion.

To study collusion, I use the notion of ``coalition-proofness'' developed by \cite{bernheim1987coalition}. This solution concept allows testing whether a protocol remains effective when players engage in non-binding pre-play communication. It allows senders to discuss their strategies before consultation, but not to make commitments. For example, managers may share their intentions with each other before filing a report to the CEO. Even though managers are unable to make credible and binding commitments, they may still be able to effectively coordinate their actions. As we shall see, coalition-proofness turns out to be crucial for the main result.

The first part of the paper establishes necessary conditions for efficiency and robustness to collusion. An absence of communication results in inaccurate decision-making, whereas single-sender arrangements always result in wasteful persuasion attempts. It follows that multi-sender protocols are necessary for efficiency.  They are also sufficient because, in the absence of pre-play communication, senders cannot coordinate persuasion, allowing the receiver to exploit this inability by privately consulting more than one sender. However, private communication protocols become susceptible to collusion if senders have the opportunity to discuss their strategies  with each other before being consulted. In such cases, senders can restore coordination, undermining the protocol's effectiveness. No private protocol can simultaneously achieve efficiency and resilience to collusion.

Public protocols prescribe consulting senders through a sequential and public procedure. These arrangements are neither efficient nor robust to collusion when senders have relatively aligned preferences over decision-making. In these cases, there are always contingencies where senders can coordinate persuasion through pre-play communication.

The second part of this paper focuses on the last type of minimal arrangement left to analyze: public advocacy, that is, the sequential and public consultation of senders with conflicting interests over decision-making. The main result shows that public advocacy is efficient and robust to collusion. Importantly, it is the only minimal communication protocol to have these desirable properties. A characterization of the efficient equilibrium is provided, showing the mechanism through which the receiver achieves efficiency: the report delivered by the first speaker sets the burden of proof borne by the second speaker, who has to prove its case ``beyond a reasonable doubt.'' The endogenously determined burden of proof ensures that both senders consistently report truthfully. As a result, the receiver learns their private information and makes fully informed decisions. No resources are wasted in the attempt to persuade the receiver. All players obtain the payoff they would get if there were no information asymmetries in the first place. 

The rest of the paper is organized as follows. Section~\ref{sec:lit} reviews the related literature, and Section~\ref{sec:model} presents the model. The main results are in Section~\ref{sec:protocols}. Section~\ref{sec:disc} discusses the model's assumptions and the robustness of the results. Finally, Section~\ref{sec:conc} concludes.

\section{Literature}\label{sec:lit}

This paper contributes to the literature on organization design. In this line of work, \cite{milgrom1988employment} recognizes that influence activities constitute a direct opportunity cost for organizations. \cite{milgrom1988employment} focuses on restrictions of decision makers' discretion as a tool to limit these costs. Differently, the current paper is concerned with the design of communication protocols that eliminate influence activities. \cite{jehiel1999} and \cite{deimen2019} also study optimal information structures in organizations. In contrast with the current paper, they consider settings where there is no disagreement between players. Without an agency problem, influence activities are not a concern.

A strand of literature finds that optimal organization design results in advocacy structures. \cite{dewatripont1999advocates} study the optimal provision of incentives for information gathering. They consider settings with verifiable information and no agency problem. \cite{battaglini2002} takes a different approach by studying strategic communication with biased senders in a cheap talk framework. Both \cite{dewatripont1999advocates} and \cite{battaglini2002} make a case for ``static'' types of advocacy. By contrast, \cite{krishna2001exp} show the optimality of rebuttals in a cheap talk model where senders engage in an extended debate. The current work differs from these papers in three key aspects. First, it is a model of partially verifiable information.\footnote{As we shall see, the presence of costly communication implies that each sender $j$ can reasonably deliver only reports within a set $M_j(\theta)$, which depends on the true state $\theta$ and always includes $\theta$ itself. This approach to partially verifiable information is analogous to that in~\cite{green1986partially}.} Second, it studies settings where influence activities yield direct costs to the organization that are not informational. Third, the analysis includes a wider array of arrangements, showing that public advocacy is the only communication protocol with several desirable characteristics.\footnote{The organization's problem considered here cannot be appreciated by models where information is not partially verifiable: influence activities are impossible when information is fully verifiable, and they come at no cost when information is not verifiable as in cheap talk models.}

In this model, influence activities are costly due to misreporting costs. The signaling structure considered here relates this paper to the ``costly talk'' literature \citep{kartik2007,kartik2009,ottaviani2006}. These papers are mainly concerned with the single-sender case. \cite{vaccari2023competition} considers the case where two conflicting senders communicate simultaneously. As in traditional signaling models, the equilibria of these settings involve wasteful signaling expenditures, and thus are inefficient. Differently from all these papers, the current work considers a larger class of protocols and equilibria. The role of collusion, which is central here, is absent in this literature. 

Efficient outcomes are extremely rare in settings where information is conveyed through costly actions. In canonical signaling models, separating outcomes clear information asymmetries, but involve wasteful signaling costs \citep{spence1973jobmarketsignaling}. There are a few exceptions. \cite{bagwell1991} study a limit pricing model with two incumbent firms. They show the existence of an equilibrium that yields the complete-information outcome. Likewise, \cite{emons2009accuracy} construct an efficient equilibrium in an adversarial model of costly communication.\footnote{To the best of my knowledge, those are the only two examples of efficient equilibria in settings where information is conveyed through costly actions. By contrast, efficiency is easier to obtain in cheap talk and disclosure games, where costs are not a concern by default. For example, see \cite{battaglini2002} and \cite{milgrom1986}.} Results in the current paper show that efficiency can always be obtained in a large class of games with costly communication. However, only one specific type of mechanism can reliably yield efficient outcomes.

Finally, a literature studies mechanisms to elicit correlated information \citep{cremer1988full,mcafee1992correlated}. \cite{laffont2000mechanism} do so in settings where agents can enter collusive agreements. Similarly, the current paper focuses on the elicitation of perfectly correlated information from multiple senders that can collude. Differently, it studies efficiency in environments where communication is costly, and the cost of reports is type-dependent. Due to this last feature, the revelation principle and the mechanism design approach \citep{myerson1983mechanism,myerson1986multistage} cannot be used here.\footnote{I thank Patrick Rey for this observation.} There are mechanism design papers addressing costly misreporting, such as those by \cite{deneckere2022} and \cite{kartik2012implementation}. The latter focuses on full implementation while requiring no costs in equilibrium. In this paper, the receiver is not subject to the same level of commitment assumed by the mechanism design approach.

\section{The model}\label{sec:model}

There are $N\geq 1$ senders in the set $S=\{1,\ldots, N\}$, and one receiver ($r$). Nature selects a state $\theta$ according to some distribution $F$ with density $f$ and full support in $\Theta = \mathbb{R}$. Only the senders observe the realized state. Then, depending on the setting, communication takes place either sequentially or simultaneously. In a sequential protocol, the order of communication is determined by senders' indexes $j\in S$, with sender~1 reporting first, and so on. When communicating, sender~$j\in S$ delivers a report $r_j\in\Theta$ with the literal or exogenous meaning ``the state is $\theta=r_j$.'' After observing all the senders' reports, the receiver selects an alternative $a$ in the binary set $\{a^+, a^-\}$.

\emph{Payoffs.---} Player $i\in S\cup\{r\}$ obtains a payoff of $u_i(a,\theta)$ when the receiver selects action $a$ in state $\theta$. Define player $i$'s payoff-difference in state $\theta$ as $\Delta u_i(\theta):=u_i(a^+,\theta)-u_i(a^-,\theta)$. I assume that $\Delta u_i(\theta)$ is weakly increasing in $\theta$ for every $i\in S\cup\{r\}$. Under this assumption, the state can be interpreted as a vertical differentiation parameter measuring the relative appeal of $a^+$ with respect to $a^-$. I normalize the receiver's payoffs by setting $\Delta u_r(\theta)>0$ for all $\theta>0$, $\Delta u_r(\theta)<0$ for all $\theta<0$, and $\Delta u_r(0)\geq 0$. For every sender $j\in S$, either $\Delta u_j(\theta)>0$ for some $\theta<0$, or $\Delta u_j(\theta)<0$ for some $\theta>0$. I say that two senders $j$ and $i$ are opposed-biased if $\Delta u_i(0)\cdot \Delta u_j(0)<0$, and are like-biased otherwise.

\emph{Costly talk.---} Sender $j$ incurs a finite cost $C_j(r_j,\theta)$ when reporting $r_j$ in state $\theta$. The cost function $C_j$ is continuous in its arguments and strictly increasing in $|r_j-\theta|$, with $C_j(\theta,\theta)=0$ for every $\theta\in\Theta$. That is, misreporting is increasingly costly in the size of the lie, whereas truthful reporting is costless. Moreover, $C_j(r_j,\theta)\to\infty$ as $|r_j-\theta|\to\infty$, reflecting that extreme lies are prohibitively costly. Additionally, for every $r_j\in\Theta$, $C_j(r_j,\theta)>C_j(r_j,\theta')$ if $|r_j-\theta|>|r_j-\theta'|$. That is, reports are cheaper when delivered from closer states. Sender $j$'s total utility is $w_j(r_j,a,\theta)=u_j(a,\theta)-C_j(r_j,\theta)$.

\emph{Strategies.---} When communication is simultaneous, a pure strategy for sender~$j$ is a function $\rho_j:\Theta\to\Theta$. When communication is sequential, a pure strategy for sender~$j$ is a function $\rho_j:\Theta^{j}\to\Theta$. A profile of reports $\{r_j\}_{j=1}^N$ is off-path if, given the senders' strategies, $\{r_j\}_{j=1}^N$ will never be observed by the receiver. A posterior belief function for the receiver is a mapping $p:\Theta^N \to \Delta(\Theta)$ that, given the senders' reports, generates posterior beliefs $p(\theta\mid\{r_j\}_{j\in S})$. Given the senders' reports and posterior beliefs $p$, the receiver selects an action in the sequentially rational set $\beta$, where
\[
\beta(\{r_j\}_{j\in S})=\argmax_{a\in\{a^+,a^-\}} \mathbb{E}_p \left[u_r(a,\theta)\mid\{r_j\}_{j\in S}\right].
\]
I assume that the receiver selects $a^+$ when indifferent between the two alternatives.

\emph{Equilibrium.---} The equilibrium concept is the perfect Bayesian equilibrium (PBE). To test for the protocols' robustness against collusion, I use the two related concepts of strong Nash equilibrium \citep{aumann1959} and coalition-proof Nash equilibrium \citep{bernheim1987coalition}. An equilibrium is \emph{strong} if no coalition of players can jointly deviate so that all players in the coalition get strictly better payoffs. An equilibrium is \emph{coalition-proof} if it is robust to coalitional deviations that are \emph{self-enforcing}. A coalition is self-enforcing if no proper sub-coalition exists in which all its members, while treating the actions of the complement as fixed, can unanimously agree to deviate from the proposed coalitional deviation in a way that makes each member of the sub-coalition strictly better off.\footnote{The notion of coalition-proofness considers group deviations that are consistent with the model's non-cooperative framework. For a formal definition of strong and coalition-proof Nash equilibrium, see \cite{aumann1959} and \cite{bernheim1987coalition}, respectively. For a textbook definition of perfect Bayesian equilibrium, see \cite{fudenberg1991game}.} I refer to equilibria where senders always report truthfully as \emph{truthful}, and to equilibria where the receiver always learn the state as \emph{fully revealing}.

\subsection{Definitions}

Here I define concepts that are useful for the analysis that follows. A discussion of the model's assumptions ensues in Section~\ref{sec:model_disc} and \ref{sec:disc}.

\emph{Reach.---} Misreporting costs significantly affect information transmission, as certain reports cannot be profitably delivered from specific states. The notion of ``reach'' captures how far a sender can misreport before incurring a sure loss. Since senders can misreport by either exaggerating or understating the realized state, it is necessary to define two sender-specific, state-dependent thresholds---one for each direction of the lie. Reports that exceed these thresholds at a given state are unprofitable.

 I define the reach of sender~$j$ in state $\theta$ as the report which associated cost offsets $j$'s potential gains. Formally, the reach of sender $j$ in state $\theta$ is $\bar r_j(\theta):=\max\{r_j\in \mathbb{R} \text{ s.t. } |\Delta u_j(\theta)|=C_j(r_j,\theta)\}$ if $\Delta u_j(0)>0$, and it is $\ubar r_j(\theta):=\min\{r_j\in\mathbb{R} \text{ s.t. } |\Delta u_j(\theta)|=C_j(r_j,\theta)\}$ if $\Delta u_j(0)<0$. The reach is computed under the condition that the sender's report is persuasive. Intuitively, in equilibrium sender $j$ will never deliver reports higher than $\bar r_j(\theta)$ or lower than $\ubar r_j(\theta)$, as these reports are strictly dominated by truthful reporting independently of the receiver's decision.

\emph{Protocols and minimality.---} A communication protocol describes how communication between senders and the receiver takes place. I say that a protocol is \emph{minimal} if it can be described just by the senders' relative bias (i.e., opposed- or like-biased) and by the confidentiality of their reports (i.e., public or private).\footnote{Similarly, in \citet[p.~1400]{goltsman2009mediation} a protocol is ``the sequence in which the parties can talk, and whether their messages are public or private.''} Single-sender protocols, analyzed in Proposition~\ref{prop:one}, do not need a description of relative bias and confidentiality. Table~\ref{tab:proto} illustrates all minimal multi-sender protocols and refers to their respective results. Section~\ref{sec:disc} discusses more complicated protocols.

I refer to protocols that are not minimal as \emph{complicated}. Complicated arrangements can prescribe, for example, third-party mediation, arbitration, rebuttals, or a complex mix of private and public communication. There are three reasons why the paper emphasizes minimal protocols. First, to avoid redundancies. Section~\ref{sec:disc} discusses how complicated protocols that yield robust and efficient outcomes can be simplified to a minimal protocol while preserving robustness and  efficiency. The second reason is a cost-based rationale. In complicated protocols, the receiver \emph{must} scrutinize at least three reports. Results in Section~\ref{sec:protocols} show that efficiency can be obtained with two senders only, each delivering a single report. Third, minimality is appealing because of its simplicity. In light of this discussion and the existence result in Proposition~\ref{prop:eqm}, an analysis of complicated protocols, while potentially interesting, is necessarily of second-order importance.

More formally, a protocol with $N\geq 1$ senders is minimal if it is fully characterized by
\begin{itemize}[noitemsep,topsep=0pt]
	\item[i)] {The \emph{bias configuration}:} the protocol specifies the bias relationship, either like- or opposed-biased, for all pairs of senders $i,j\in S=\{1,\ldots, N\}$, with $i\neq j$;
	\item[ii)] {The \emph{confidentiality}:} communication is either simultaneous (i.e., private) or sequential (i.e., public). In both cases, each sender reports only once.
	\begin{itemize}[noitemsep,topsep=0pt]
		\item[iia)] When communication is simultaneous, all senders report without observing the other senders' reports, and no further specification is required;
		\item[iib)] When sequential, senders report publicly one after the other. In this case, the protocol explicitly specifies the sequence in which senders communicate.
	\end{itemize}
\end{itemize}
Most results in Section~\ref{sec:protocols} hold for any $N \geq 2$ and any communication sequence, so these characteristics of minimal protocols will not be emphasized.

\emph{Efficiency.---} In an efficient equilibrium, all players obtain their respective complete-information payoffs: the receiver selects action $a^+$ when $\theta\geq 0$ and action $a^-$ otherwise; senders always report truthfully, i.e., $\rho_j(\cdot,\theta)=\theta$ for every $j\in S$ and $\theta\in\Theta$. A communication protocol is efficient if it allows for an efficient equilibrium. Since truthful equilibria are fully revealing while non-truthful ones are wasteful, an equilibrium is efficient if and only if it is truthful. Full revelation is neither necessary nor sufficient for efficiency.\footnote{Full revelation may occur alongside misreporting (see, e.g., \cite{ottaviani2006} and \cite{kartik2007}). Therefore, full revelation is not sufficient for efficiency. To obtain her complete-information payoff in this model, the receiver does not need to learn every state. She just need to know if the state is positive or not. Therefore, full revelation is not even necessary for efficiency. However, every state is fully revealed in the efficient equilibrium characterized by Proposition~\ref{prop:eqm}.}

The term ``efficiency'' used here may differ from that employed in other works, as we typically refer to Pareto or utilitarian efficiency, which do not necessarily coincide with complete-information outcomes. Nevertheless, the main result outlined in Proposition~\ref{prop:eqm} constitutes a Pareto efficient outcome. Utilitarian efficiency can be achieved by ``public advocacy'' protocols provided that the receiver's gains from accurate decision-making sufficiently outweigh the senders' gains from persuasion. This assumption is reasonable for most organizations and institutions.

\begin{table}[htbp]
	\centering
	\begin{tabular}{lp{6cm}p{6cm}}
		\toprule
		\multicolumn{1}{c}{}&
		\multicolumn{2}{c}{{\bf Biases}}\\
		\cmidrule(lr){2-3}    
		{\bf Confident.} & \centering{Like-biased} & \hspace{1.5cm} {Opposed-biased}\\
		\midrule
		Private     &  Proposition~\ref{prop:sim} -- There are no coalition-proof and efficient PBE & Proposition~\ref{prop:sim} -- There are no coalition-proof and efficient PBE  \\ 
		Public 		&  Proposition~\ref{prop:like} -- There are no coalition-proof and efficient PBE & Proposition~\ref{prop:eqm}  and Corollary~\ref{cor:strong} -- There exist efficient and strong PBE; characterization provided\\
		\bottomrule
	\end{tabular}
	\caption{Minimal multi-sender protocols and related results.}
	\label{tab:proto}
\end{table}

\subsection{Discussion of the model}\label{sec:model_disc}

Before proceeding with the analysis, it is worth making a few remarks about the model's structure. First, for simplicity, it is assumed that the state space coincides with the real line. However, for the results to carry through, it is sufficient that the state space is large enough to allow the delivery of reports that unambiguously signal the state's sign. Specifically, it is sufficient that $\Theta\supseteq \left[\min\{\ubar r_j(0)\}_{j\in S}, \max\{\bar r_j(0)\}_{j\in S}\right]=\bar\Theta$. Intuitively, there is no equilibrium where senders misreport by delivering reports that are strictly beyond their own reach at zero. Every such a report is always equilibrium dominated.\footnote{By definition of reach, reports $r_j>\bar r_j(0)$ (resp.~$r_j<\ubar r_j(0)$) are strictly dominated when the realized state is negative (resp.~positive). They are unambiguous signals of the state's sign. In some equilibria, reports that are equal or close to a sender's reach at zero are delivered on-path. The state space size should allow for the delivery of such reports. For a detailed discussion, see the analysis in \cite{vaccari2023competition}.}  It is qualitatively irrelevant for the paper's results if the state space's size is larger than $\bar\Theta$.

Second, the model assumes that payoff-differences $\Delta u_i(\theta)$ are increasing in $\theta$. The results do not hinge on this assumption. All findings hold true as long as the payoffs are such that: (i) the sets of states in which a player prefers a particular alternative are convex; (ii) each report remains cheaper when delivered from states that are closer to the report. More formally, for every player $j$, the sets $\Theta_j^{+}=\{\theta\in\Theta \;|\; \Delta u_j(\theta)\geq 0\}$ and $\Theta_j^{-}=\{\theta\in\Theta \;|\; \Delta u_j(\theta)\leq 0\}$ must be convex. In addition, for every sender~$j$, report $r_j\in\Theta$, action $a\in\{a^+,a^-\}$, and pair of states $\theta',\theta''$ such that $|r_j-\theta'|<|r_j-\theta''|$, the payoffs must satisfy $w_j(r_j,a,\theta')>w_j(r_j,a,\theta'')$. Together, these conditions maintain the results while allowing $\Delta u_j(\theta)$ to be sometimes decreasing in its argument.\footnote{The analysis in Section~\ref{sec:protocols} requires that, for every sender, the reaches are finite, unique, and monotonic in the state. The misreporting cost functions $C_j$ do not need to be everywhere continuous as long as the reaches are well-defined for every sender.}

Two benchmarks are relevant. First, if we set $C_j(r_j,\theta)=0$ for every $r_j,\theta\in\Theta$ and $j\in S$, the model becomes one of cheap talk. In this case, every coalition-proof equilibrium results in babbling in states where some senders have a conflict of interest with the receiver.\footnote{In the cheap talk version, credible information transmission remains possible in equilibrium, e.g., in extreme states where all players agree on the best course of action.} Senders can credibly transmit information only when recommending the action they are biased against. As a result, there are no efficient and robust protocols. Second, misreporting is prohibitively expensive when $C_j(r_j,\theta)>|\Delta u_j(\theta)|$ for every $r_j\neq\theta$ and $j\in S$. In this case, every communication protocol results in an efficient and coalition-proof outcome. This paper is concerned with the intermediate case of partially verifiable information where misreporting is possible at a non-prohibitive cost.

Furthermore, it is assumed that the senders' report space coincides with the real line. Results hold under more general and abstract specifications of the senders' report space and cost function, including the possibility of vague language and information withholding. This extension and several other variations of the main model are discussed in Section~\ref{sec:disc} and Appendix~\ref{app:app}.

\section{Communication protocols}\label{sec:protocols}

The end goal of the first part of this section is to rule out communication protocols that are either inefficient or not robust to senders' collusion. A protocol with no senders is not efficient because it involves decision-making under risk. The following proposition shows that even single-sender protocols are not efficient.
\begin{proposition}[Single-sender protocols]\label{prop:one}
	There are no efficient PBE if $N=1$.
\end{proposition}
\begin{proof}
	Consider a protocol where $N=1$ and $\Delta u_1(0)>0$ (the proof is similar for $\Delta u_1(0)<0$). Suppose that there exists an efficient PBE where $\rho_1(\theta)=\theta$ for every $\theta\in\Theta$. The receiver's posterior beliefs are degenerate on $\theta=r_1$. Sender~1 can profitably deviate by delivering $r_1\in[0,\bar r_1(\theta'))$ in some $\theta'<0$ where $\bar r_1(\theta')>0$, contradicting the existence of a PBE in truthful strategies.
\end{proof}
Single-sender protocols are inefficient because there are always states in which the consulted sender misrepresents information. With no other source of advice, the receiver cannot cross-validate reports to spur truthful reporting. Proposition~\ref{prop:one} implies that multi-sender protocols are necessary for efficiency. The following proposition considers arrangements where the receiver privately or simultaneously consults multiple senders with any bias configuration.

\begin{proposition}[Simultaneous communication]\label{prop:sim}
	Consider protocols with $N\geq 2$ senders that communicate simultaneously. Efficient PBE of these protocols are not coalition-proof.
\end{proposition}
\begin{proof}
	Consider first the case where all senders are like-biased, i.e., $\Delta u_j(0)>0$ for all $j\in S$ (the proof is similar if $\Delta u_j(0)<0$ for all $j\in S$). There exists an efficient PBE where $\rho_j(\theta)=\theta$ for every $j\in S$ and $\theta\in\Theta$, and beliefs $p$ are such that $\beta(\{r_j\}_{j\in S})=a^-$ if $r_i\neq r_k$ for some $i, k \in S$. Given strategies and beliefs, no sender has an incentive to deviate from truthful reporting. Posterior beliefs are pinned down by Bayes' rule only for the case $r_i=r_j$ for all $i,j\in S$. When all reports are identical, the receiver assigns probability $1$ to $\theta=r_j$. Off-path beliefs ensure that individual deviations are not profitable. 
	
	However, efficient equilibria of this configuration are not coalition-proof. Consider a state $\vartheta_\epsilon<0$ such that $\min\{\bar r_j(\vartheta_\epsilon)\}_{j\in S}>0$. There is a coalition formed by all senders in $S$ such that, when the state is $\vartheta_\epsilon$, each sender~$j\in S$ deviates from truthful reporting by delivering $r_j=r'\in[0,\min\{\bar r_j(\vartheta_\epsilon)\}_{j\in S})$. Upon observing $\{r_j=r'\}_{j\in S}$, the receiver selects action $a^+$. Given $p$, if some sender delivers a report different than $r'$, then the receiver selects $a^-$. Therefore, this coalitional deviation is mutually beneficial and self-enforcing: there is no proper sub-coalition that, taking fixed the action of its complement, can agree to deviate from the original deviation in a way that makes all of its members better off. As a result, every efficient PBE of this like-biased configuration is not coalition-proof.
	
	Consider now the case where at least two senders are opposed-biased. That is, there are at least two senders $i,j\in S$, $i\neq j$, such that $\Delta u_i(0)\cdot \Delta u_j(0)<0$. First, I show that misreporting occurs in every PBE when the number of senders is $N=2$. Suppose there exists a PBE where misreporting never occurs, that is, where $\rho_1(\theta)=\rho_2(\theta)=\theta$ for every $\theta\in\Theta$. Consider such a truthful equilibrium, two opposed-biased senders with $\Delta u_2(0)<0< \Delta u_1(0)$, and a state $\theta=\vartheta_\epsilon>0$, where $\vartheta_\epsilon$ is small enough. We have that $\rho_1(\vartheta_\epsilon)=\rho_2(\vartheta_\epsilon)=\vartheta_\epsilon$. To discourage deviations, off path beliefs must be such that $\beta(\vartheta_\epsilon,-\vartheta_\epsilon)=a^+$. However, there always exists an $\vartheta_\epsilon>0$ such that, when the state is $\theta=-\vartheta_\epsilon$, sender~1 can profitably deviate from the prescribed truthful strategy (i.e, $\rho_1(-\vartheta_\epsilon)=\rho_2(-\vartheta_\epsilon)=-\vartheta_\epsilon$) by reporting $r_1=\vartheta_\epsilon$, as $\Delta u_1(-\vartheta_\epsilon)> C_1(\vartheta_\epsilon,-\vartheta_\epsilon)$. This contradicts that there exists an equilibrium where misreporting never occurs for the case $N=2$.
	
	For the case $N>2$, efficient PBE can exist but are not coalition-proof. Consider protocols with more than two senders that communicate simultaneously. Define the sets of senders $Z=\{j\mid\Delta u_j(0)>0\}$ and $Y=\{i\mid\Delta u_i(0)<0\}$, and the profile of reports $\tilde r_Z=\{r_j\}_{j\in Z}$ and $\tilde r_Y=\{r_i\}_{i\in Y}$. Say that $\tilde r_L = x$ when $r_j=x$ for all $j\in L$. The receiver's decision rule is $\beta(\tilde r_Z, \tilde r_Y)$. In a truthful equilibrium, $\tilde r_Z =\tilde r_Y=\theta$ for all $\theta\in\Theta$. Given a state $\theta=-\vartheta_\epsilon<0$, beliefs are such that $\beta(-\vartheta_\epsilon,-\vartheta_\epsilon)=a^-$ and $\beta(\vartheta_\epsilon,\vartheta_\epsilon)=a^+$. Suppose that off path beliefs are such that $\beta(\vartheta_\epsilon,-\vartheta_\epsilon)=a^+$, and take $\vartheta_\epsilon$ small enough\footnote{If instead beliefs are such that $\beta(\vartheta_\epsilon,-\vartheta_\epsilon)=a^-$, then the proof is similar by considering a state $\theta=\vartheta_\epsilon>0$ small enough and a deviation from the coalition of senders $i\in Y$.} so that $\bar r_j(-\vartheta_\epsilon)\geq \vartheta_\epsilon$ for every $j\in Z$. Denote by $\tilde r_Z^*=\{r_j^*\}_{j\in Z}$ a profile of reports with the following features: $\beta(\tilde r_Z^*,-\vartheta_\epsilon)=a^+$, and $\beta(\{r_j\}_{j\in Z},-\vartheta_\epsilon)=a^-$ for any $\{r_j\}_{j\in Z}$ such that $C_j(r_j,-\vartheta_\epsilon)\leq C_j(r_j^*,-\vartheta_\epsilon)$ for every $j\in Z$, with a strict inequality for some $j$. The report $\tilde r_Z^*$ exists because $\beta(\vartheta_\epsilon,-\vartheta_\epsilon)=a^+$, and therefore $r_j^*\leq \vartheta_\epsilon$ for every $j\in Z$. There is a coalition formed by all senders in $Z$ such that, when the state is $\theta=-\vartheta_\epsilon$, each $j\in Z$ can deviate to $r_j^*\in \tilde r_Z^*$. This deviation is mutually beneficial and self-enforcing: by construction of $\tilde r_Z^*$, there is no proper sub-coalition that, taking fixed the action of its complement, can agree to deviate from the deviation in a way that makes all of its members better off. Equilibria in truthful strategies of these protocols are not coalition-proof \citep{bernheim1987coalition}.
\end{proof}

There are two reasons why simultaneous communication protocols may initially seem a promising avenue toward efficiency. First, fully revealing equilibria in truthful strategies exist in simultaneous cheap talk games with three or more senders with any bias type.\footnote{See, for example, the introductory section in \cite{battaglini2004}. In these equilibria, individual deviations from truthful reporting are immediately detected. The same intuition carries over to the costly talk signaling structure considered here.}  Second, the receiver can achieve efficiency even by simultaneously consulting two (or more) like-biased senders. In this last case, the receiver can induce truthful reporting by applying skepticism when reports do not coincide. Simultaneous communication protocols can exploit senders' lack of coordination to make any individual persuasion attempt futile. As we have seen, collusion effectively restores senders' ability to coordinate persuasion.

Proposition~\ref{prop:sim} shows that simultaneous protocols are ineffective when senders can engage in non-binding pre-play communication. The possibility of discussing strategies before consultation allows senders to coordinate persuasion in a self-enforcing way, even though they have no commitment power.  This last result implies that sequential communication protocols are necessary to jointly achieve efficiency and robustness to collusion. The following proposition considers arrangements where the receiver sequentially consults multiple like-biased senders.

\begin{proposition}[Sequential consultation of like-biased senders]\label{prop:like}
	There are no efficient and coalition-proof PBE if there are $N\geq 2$ like-biased senders that communicate publicly and sequentially.
\end{proposition}
\begin{proof}
	Suppose there is a PBE in truthful strategies where $\Delta u_j(0)>0$ for all $j\in S$ (the proof is analogous if $\Delta u_j(0)<0$ for all $j\in S$). In this equilibrium, $\rho_1(\theta)=\rho_n(\{r_j=\theta\}_{j=1}^{n-1},\theta)=\theta$ for every $\theta\in\Theta$ and $n\in\{2,\ldots,N\}$. On-path, the receiver's beliefs are degenerate on $\theta=r_j$, $j\in S$. Consider a state $\theta'<0$ and a report $r'\geq 0$ such that $\Delta u_j(\theta')>0$ and $\bar r_j(\theta')>r'\geq 0$ for all $j\in S$. An argument similar to that for simultaneous like-bias protocols (Proposition~\ref{prop:sim}) shows that truthful PBE of this configuration are not coalition-proof. Through pre-play communication, senders can agree to deliver $r'$ in state $\theta'$ provided that all senders before delivered $r'$ as well. Upon observing $\{r_j=r'\}_{j\in S}$, the receiver selects action $a^+$. Given $p$, if some sender delivers a report different than $r'$, then the receiver selects $a^-$. After observing the tuple of reports $\{r'\}_{j=1}^{N-1}$, sender~$N$ can profitably deliver $r_N=r'$, as $\beta(\{r'\}_{j=1}^N)=a^+$. By induction, the same is true for every sender $j\in S$. Therefore, this coalitional deviation is mutually beneficial and self-enforcing: there is no proper sub-coalition that, taking fixed the action of its complement, can agree to deviate from the original deviation in a way that makes all of its members better off. Truthful equilibria of this configuration are not coalition-proof.
\end{proof}

Sequential protocols with like-biased senders remain problematic because collusion can still be enforced through non-binding pre-play communication. Before consultation, senders can jointly agree to deviate from truth-telling to achieve persuasion. It follows that truthful equilibria and efficient outcomes cannot be supported by sequential communication protocols where all senders are like-biased. Next, we turn to analyze sequential protocols with opposed-biased senders.  

\subsection{Public advocacy}\label{sec:public_advocacy}

Altogether, Propositions~\ref{prop:one}, \ref{prop:sim}, and \ref{prop:like} rule out a large class of communication protocols. This first batch of results implies that efficient and robust outcomes may be achieved only through a specific arrangement: multi-sender protocols where communication is sequential and senders are opposed-biased. I hereafter refer to these arrangements as \emph{public advocacy}, as they involve the sequential and public consultation of senders with conflicting interests. The following proposition shows that public advocacy yields efficient equilibria.

\begin{proposition}[Public advocacy]\label{prop:eqm}
	Consider a protocol with $N=2$ senders that communicate sequentially and are opposed-biased, i.e., $\Delta u_2(0)<0<\Delta u_1(0)$. There exists an efficient perfect Bayesian equilibrium where
	\begin{equation*}
		\rho_1(\theta)=\theta \quad \forall \; \theta \in \Theta,
	\end{equation*}
	\begin{equation*}
		\rho_2(r_1,\theta) =
		\begin{cases}
			\min\{ \ubar r_2(0),\theta \} & \quad \text{if } \; \theta<0\leq r_1,\\
			\theta & \quad \text{otherwise. }
		\end{cases}
	\end{equation*}
	\begin{equation*}
	p \text{ are s.t. } \beta(r_1,r_2)=
	\begin{cases}
		a^+ & \quad \text{if } \; r_1\geq 0 \; \text{ and } \; r_2>\ubar r_2(0), \\
		a^- & \quad \text{otherwise. }
	\end{cases}
\end{equation*}
\end{proposition}
\begin{proof}
	Given beliefs $p$ and sender~2's strategy, truthful reporting is strictly dominant for sender~1 in every state. Given beliefs $p$ and sender~1's strategy: (i) if $r_1<0$, then truthful reporting is always strictly dominant for sender~2, as $\beta(r_1,\cdot)=a^-$ for every $r_2$; (ii) if $r_1 \geq 0$ and $\theta\geq 0$, then truthful reporting is strictly dominant for sender~2 because action $a^-$ can be induced only by delivering a $r_2\leq \ubar r_2(0)$ which, by definition of reach, is never profitable in non-negative states; (iii) if $r_1 \geq 0$ and $\theta<0$, then sender~2 can induce action $a^-$ only by delivering some $r_2\leq \ubar r_2(0)$. By definition of reach, $r_2=\min\{\ubar r_2(0),\theta\}$ is strictly dominant in this case. Given senders' strategies, beliefs $p$ are pinned down by Bayes' rule only for $r_1=r_2$, and are free otherwise. When reports are identical, the receiver assigns probability $1$ to $\theta=r_1=r_2$. Off-path beliefs are free and set as in Proposition~\ref{prop:eqm}. Since senders play truthful strategies, have no profitable individual deviations, and beliefs are according to Bayes' rule whenever possible, this is an efficient equilibrium.
\end{proof}

Proposition~\ref{prop:eqm} provides an equilibrium characterization, which allows us to understand the mechanism supporting truthful reporting on the equilibrium path. The key to efficiency in public advocacy stands in how the receiver allocates the burden of proof between the two senders.\footnote{Since the receiver fully learns the state after sequentially consulting two opposed-biased senders, efficiency can be achieved by public advocacy protocols with $N\geq 2$ senders. The focus on $N=2$, based on a minimality principle, is therefore without loss of generality.} Beliefs must be consistent with Bayes' rule when senders deliver identical reports in a truthful equilibrium. In these cases, the receiver always follows the senders' recommendations. By contrast, beliefs are free in all those cases where senders disagree. The construction of suitable off path beliefs is crucial in sustaining truthful equilibria.

To illustrate the role of beliefs, consider the case where sender~1 prefers $a^+$ and sender~2 prefers $a^-$. Recall that sender~1 speaks first.\footnote{The order in which senders communicate is irrelevant for the result in Proposition~\ref{prop:eqm}.} When sender~1 recommends its least favorite action $a^-$ (i.e., $r_1<0$), the receiver selects $a^-$ no matter what sender~2 reports. By contrast, when sender~1 suggests its favorite action $a^+$ (i.e., $r_1\geq 0$), the receiver's decision depends on what sender~2 reports. In this case, sender~2 can convince the receiver to select $a^-$ only by delivering a report dominated in non-negative states (i.e., $r_2\leq \ubar r_2(0)$). Intuitively, the receiver follows sender~1's advice to choose $a^+$ only if sender~2 fails to provide undeniable evidence that the state is negative. The burden of proof allocation ensures that senders have no incentive to deviate from truthful reporting.\footnote{As set in Proposition~\ref{prop:eqm}, the burden of proof is reminiscent of the ``onus'' applied in trials before tribunals. In legal disputes, one party is initially presumed to be correct. In contrast, the other party is burdened by providing sufficiently persuasive evidence to prove its case ``beyond a reasonable doubt.''}

The sequential structure of public advocacy is key for efficiency. To illustrate the significance of sequentiality, it is useful to compare the equilibrium in Proposition~\ref{prop:eqm} with simultaneous protocols involving two opposed-biased senders as studied in Proposition~\ref{prop:sim}. We have observed that a truthful equilibrium where $\rho_1(\theta)=\rho_2(\theta)=\theta$ for every $\theta\in\Theta$ cannot be sustained by a simultaneous protocol because there are always states where one of the two senders can profitably deviate from the prescribed strategy. By contrast, the sequential protocol in Proposition~\ref{prop:eqm} sustains truthful reporting on the equilibrium path, that is, $\rho_1(\theta)=\rho_2(\theta,\theta)=\theta$ for every $\theta\in\Theta$. Crucially, this is possible because whoever speaks second can condition their report on the one delivered by the first speaker. Given the receiver's beliefs, the second sender has always the option to profitably counteract any persuasion attempt by the first sender. At the same time, whoever speaks second cannot profit from persuasion, as it is always prohibitively expensive.\footnote{One might think that we can use the receiver's beliefs as outlined in Proposition~\ref{prop:eqm} to attain truthful equilibria in simultaneous protocols. However, this approach would prove ineffective: in a slightly negative state, say $\theta=-\vartheta_\epsilon$, sender~2 would be compelled to misreport and deliver some $r_2\leq\ubar r_2(0)<-\vartheta_\epsilon$ to correctly induce $\beta=a^-$. This strategy is inherently non-truthful, and thus yields inefficiencies.}

The adversarial structure of public advocacy provides an additional benefit: the senders, having conflicting goals, cannot coordinate to influence the receiver's decision. Resilience to collusion is desirable in organizations where informed agents can discuss their intentions before being consulted by the receiver. The following corollary confirms that the protocol in Proposition~\ref{prop:eqm} is robust to non-binding pre-play communication.

\begin{corollary}\label{cor:strong}
	The PBE in Proposition~\ref{prop:eqm} is strong.
\end{corollary}
\begin{proof}
	The proof follows from the observation that there is no profitable coalitional deviation involving two opposed-biased senders. First, any deviation from the prescribed equilibrium entails a cost for at least one sender. Second, a coalitional deviation that makes one sender strictly better off must make the other sender strictly worse off, as they are opposed-biased. Furthermore, the receiver cannot gain from a coalitional deviation because the equilibrium is already efficient. Therefore, the equilibrium in Proposition~\ref{prop:eqm} is strong \citep{aumann1959} and coalition-proof \citep{bernheim1987coalition}.
\end{proof}

\section{Discussion}\label{sec:disc}

Findings in the previous section show that there is a unique communication protocol that is efficient, minimal, and resilient to collusion. This protocol, called public advocacy, requires the receiver to consult sequentially and publicly two senders with conflicting interests. This section discusses the robustness of these findings, and highlights further differences with respect to related work.

The model assumes that senders can neither withhold information nor deliver vague reports. The results do not rely on these limitations. Consider an extension of the model where senders can omit or muddle information at a cost.\footnote{Vagueness and omissions are easily detectable, whereas misreporting is not. Organizations can replace or decrease the budget of managers known to be perfectly informed and yet purposefully refuse to provide accurate information.} Every equilibrium of the main model is also an equilibrium of this extended model where the receiver interprets unexpected omissions or vagueness unfavorably. Inefficient protocols remain inefficient in this extended model with additional costly actions, as taking a costly action is by itself inefficient. Public advocacy remains the only efficient, minimal, and robust protocol.

Relatedly, the model studies a setting where the senders' report space is the real line. Results may not hold under different report spaces and cost functions. For example, Section~\ref{sec:model_disc} discusses the case where $C_j(r_j,\theta)>|\Delta u_j(\theta)|$ for every $r_j\neq\theta$ and $j\in S$. In this case, misreporting would be prohibitively expensive, and truthful reporting would always ensue. Likewise, efficient communication trivially occurs in the extreme case where the senders' report space is restricted to the costless $R_j(\theta)=\{\theta\}$ for every $\theta\in\Theta$. For this reason, there is a limited interest in incorporating message spaces and/or cost functions in the definition of communication protocols, especially in light of the existence result in Proposition~\ref{prop:eqm}. Appendix~\ref{app:report_cost} indicates that all results hold under more general and abstract specifications of the senders' report space. There, an example shows that results carry over settings where the senders' report space contains three reports only. More broadly, the appendix shows that the findings in Section~\ref{sec:protocols} extend to a wider range of report spaces and associated misreporting costs than those outlined in Section~\ref{sec:model}. 

The analysis is focused on binary decision problems, which are prevalent in numerous social and economic contexts extensively examined in the literature: e.g., voting, judicial decision-making, market entry decisions.\footnote{Other examples of widely studied binary decision settings are: medical treatment choices (e.g., opting for surgery or not), hiring decisions, technology adoption (e.g., whether to implement a new technology), and consumer purchase decisions (e.g., whether to buy a product).} It is important to recognize the implications on the results when extending the model beyond the binary-action framework. \cite{emons2009accuracy} offer an example indicating that the paper's findings are also applicable to settings with a continuous action space.\footnote{See Section~4.2, p.~147 therein. In a similar model of costly communication with a continuous action space, \cite{emons2009accuracy} show that there exist a truthful equilibrium when two biased senders with opposed goals communicate sequentially. Even though it is not discussed, their equilibrium is naturally coalition-proof because of the conflict of interest between senders.} However, an element of \emph{limited liability}, which holds naturally under the binary-action framework, is crucial for the uniqueness result. Appendix~\ref{app:more_actions} extends the analysis to a three-action framework. It shows that the presence of a third, ``safe'' alternative can yield efficient outcomes under private advocacy protocols.  This extension suggests that, in frameworks with more than two actions, a ``liability condition'' is necessary to preserve the uniqueness result outlined in Section~\ref{sec:protocols}. Without liability constraints and with more than two actions, the appendix shows that two opposed-biased senders suffice to achieve efficiency, regardless of whether they are consulted simultaneously or sequentially.

In the model, the receiver cannot implement transfers or choose the senders' payoff structure.\footnote{Mechanisms that involve transfers are inefficient because, compared to the outcome under complete information, at least one player incurs a cost when participating in a transfer. Therefore, allowing for transfers does not affect the paper's results.} Even if the receiver could affect the senders' payoff structure within the limits prescribed in Section~\ref{sec:model}, the only way to obtain an efficient and collusion-proof outcome is by using public advocacy. If the receiver could, then it would set an environment with prohibitively high misreporting costs to spur truthful reporting. However, organizations may be limited when choosing among mechanisms, either because of exogenous constraints or commitment problems. The main result is positive: efficiency can be obtained even when organizations can only decide how to structure communication.

Two additional extensions are worth discussing. First, the model assumes that senders know the state perfectly. A model variant with imperfectly informed senders introduces an information aggregation problem alongside the strategic problem of information elicitation. The receiver's need to aggregate information from imperfectly informed agents typically makes protocols with a higher number of senders more appealing. Second, it is assumed that truthful reporting is costless. However, players may incur substantial consultation costs, making protocols with fewer senders more appealing. While these extensions are relevant in broader contexts, they play no role in the present analysis, as the model's notion of efficiency presumes informed decision-making and costless reporting. In most cases, the receiver cannot make informed decisions when senders lack knowledge of the state, and communication involving costly reports necessarily entails inefficiencies.

In minimal protocols, senders have the same relative bias configuration, they speak only once, and they are subject to the same type of confidentiality. Naturally, these rules do not have to apply to more general arrangements. I refer to non-minimal protocols as \emph{complicated}. Public advocacy is one among many efficient and coalition-proof protocols when considering also complicated ones. For example, we can construct a multi-stage protocol where public advocacy is used in the first stage. Knowing that the public advocacy stage induces truthful reporting, senders communicating in all other stages optimally economize in costs by reporting truthfully. Likewise, they cannot form profitable and self-enforcing deviating coalitions. This class of complicated protocols is large, efficient, and robust to collusion. Similarly, the paper does not analyze protocols with rebuttals, that is, where senders report more than once. The main reason for omitting rebuttals is that they are unnecessary to achieve efficiency.\footnote{Introducing rebuttals would require additional and non-trivial assumptions on the misreporting costs. For example, if such costs are purely psychological, then repeating the same lie may be more costly than lying once. By contrast, direct manipulation costs such as effort and time may not duplicate.} This is in contrast with \cite{krishna2001exp}'s result that rebuttals are necessary for full revelation in their cheap talk setup.\footnote{See \citet[Proposition~1, p.~756]{krishna2001exp}, where full revelation is necessary for efficiency.} The focus on minimal protocols avoids this type of redundancies with the idea that, all else equal, organizations prefer using simpler to complicated protocols.\footnote{Proposition~\ref{prop:sim} applies also to complicated private protocols where some pairs of senders are like-biased and others are opposed-biased. Similarly, Proposition~\ref{prop:eqm} and Corollary~\ref{cor:strong} equally apply to complicated public advocacy protocols with more than two senders, of which at least two are like-biased.}

There are other differences between the results obtained here and those in \cite{krishna2001exp}'s cheap talk model, which explores a more complex setup and addresses different questions. Most importantly, while this paper examines both information transmission and influence costs, the latter---captured by the sum $\sum_{i\in N}C_j(r_j,\theta)$---are not a focus in \cite{krishna2001exp}, as such costs are set to zero in cheap talk frameworks. This distinction highlights how the intuition for the main result in this paper cannot be similar to that in \cite{krishna2001exp}, as efficiency in public advocacy arises from the ability to deliver reports that, due to their costs, serve as unequivocal signals of the state's sign.\footnote{Intuitively, the introduction of misreporting costs should (weakly) increase both information transmission and expenditures in influence activities. Public advocacy is never efficient when talk is cheap because not enough information is transmitted in equilibrium.} Both approaches contribute to understanding the interplay between communication structures and outcomes in strategic settings.\footnote{Other differences are in the protocols examined and  results obtained. \cite{krishna2001exp} show that their rebuttal protocol can fully reveal information when senders are sufficiently aligned. By contrast, this paper demonstrates that efficiency can persist even under state-independent sender preferences, extending the analysis to scenarios where senders' biases are extreme.}


\section{Concluding remarks}\label{sec:conc}

This paper shows that there exists a simple institution and a signaling structure that can solve information asymmetries at no cost. At least since \cite{akerlof1970}, it is well known that the presence of asymmetric information can yield inefficient outcomes. Subsequent work by \cite{spence1973jobmarketsignaling} shows that signaling can resolve information asymmetries but does not eliminate inefficiencies, as signals are wasteful and unproductive. This paper departs from the canonical one-sender-one-receiver setting and uses a costly talk signaling structure. When considering multi-sender protocols, the concern for collusion naturally emerges. Yet, in this setting, it is possible to structure communication in a way that fully restores efficiency without using wasteful signaling expenditures or commitment power. Likewise, there is no need for mediation, arbitration, or other complex arrangements.

The main result has potentially significant implications for the understanding of organizational design. It shows that only one minimal protocol can achieve efficiency under the threat of senders' collusion. This protocol prescribes the sequential and public consultation of two informed agents with conflicting interests. The proposed arrangement has a plain structure and does not require commitment power as ex-ante, and in the interim, the organization adheres to the protocol. Importantly, such an arrangement \emph{always} yields an efficient outcome for any configuration  permitted by the model described in Section~\ref{sec:model}.\footnote{Appendix~\ref{app:app} shows that the results extend to more general setting than the one presented in Section~\ref{sec:model}.} This finding provides a rationale for using public advocacy structures. 

Public advocacy is widespread. As an example, consider the justice system. Trials take place with an adversarial procedure of judicial decision-making, whereby two advocates---prosecutors and defendants---engage in public debates in the courtroom. Other examples of public advocacy include, e.g., managers and ministries competing for budget allocation. Budgeting processes are typically sequential and adversarial. When the contending parties disagree, the receiver---judge, CEO, or prime minister---adjudicates.

The model has two distinguishing features central to the main result. First, it considers organizations which dislike wasteful influence activities. Public advocacy is uniquely optimal only if this is the case. Fully revealing equilibria exist in simultaneous protocols, and full revelation is sufficient for optimality when the organization does not incur direct influence costs. Second, senders can engage in non-binding, pre-play communication. Public advocacy is not the uniquely efficient protocol when collusion is not a concern: the proof of Proposition~\ref{prop:sim} shows that, in these cases, simultaneous and like-bias protocols can also yield efficient outcomes. It is reasonable to jointly assume a distaste for wasteful activities and concern for collusion for a wide range of organizations and central planners.


\appendix
\section{Appendix}\label{app:app}

\subsection{Different report spaces and cost functions}\label{app:report_cost}

This section discusses different specifications of the senders' report space. For the results to hold, it is crucial that misreporting costs are not prohibitively expensive, as already discussed in Section~\ref{sec:model_disc}. Propositions~\ref{prop:one} to \ref{prop:like} require the existence of a state and a set of affordable reports such that all senders with an incentive to persuade the receiver can do so by misreporting collectively. For example, the proofs for Propositions~\ref{prop:one} to \ref{prop:like} can carry through when the senders' report space coincides with the action space.\footnote{For every sender $j\in S$, set $R_j=\{a^+,a^-\}$, $C_j(a^+,\theta)=0$ for all $\theta\geq 0$, $C_j(a^-,\theta)=0$ for all $\theta<0$, and $0<C_j(\cdot,\theta')<|\Delta u_j(\theta')|$ for all $\theta'\in\Theta$ such that $\Delta u_j(\theta')\neq0$.}

More formally, suppose that the senders' report space, $R$, is a partition of the real line such that $0=\min r_j$ for some $r_j\in R$, $C_j(r_j,\theta)=0$ if and only if $\theta\in r_j$, and $C_j(r_j,\theta)>0$ otherwise. Reports are intervals, but the report space can be more abstract as long as we can associate an interval of the real line to each report.\footnote{The mapping states-reports occurs naturally through the cost function $C_j$. For example, if $C_j(r_j,\theta)=C_j(r_j,\theta')$ for some $\theta\neq \theta'$, then we can say that $\theta,\theta'\in r_j$. To disentangle positive from negative states, assume that the cost function $C_j$ is such that there is no $\theta<0$ and $\theta'\geq 0$ such that $C_j(r_j,\theta)=C_j(r_j,\theta')$.} The proofs of Propositions~\ref{prop:one} to \ref{prop:like} directly apply to this case as long as misreporting is sufficiently affordable in some cases. Specifically, define the set of senders that want to persuade the receiver in state $\theta$ as $S'(\theta)$. Deviations from truth-telling considered by the proofs in some state $\theta'$ must have an associated cost that is lower than $|\Delta u_j(\theta')|$ for all $j\in S'(\theta')$.

By contrast, the proof of Proposition~\ref{prop:eqm} requires the existence of reports that are prohibitively expensive (i.e., dominated) in certain states. In addition to the report space as described in the previous paragraph, we need a report $r_2^-$ such that $C_2(r_2^-,\theta)>|\Delta u_2(\theta)|$ for all $\theta\geq 0$, and $C_2(r_2^-,\theta)<|\Delta u_2(\theta)|$ for all $\theta\in(t,0)$, where $t<0$ is low enough.\footnote{Low enough means  $C_1(r_1,\theta)>|\Delta u_1(\theta)|$ for all $\theta\leq t$ and all $r_1$ such that $\theta'\in r_1$ for some $\theta'\geq 0$.} Denote by $r_j^\theta$ the report such that $\theta\in r_j^\theta$, and thus $C_j\left(r_j^\theta ,\theta\right)=0$ for all $\theta\in\Theta$ and $j\in\{1,2\}$. The equilibrium in Proposition~\ref{prop:eqm} can be rewritten as follows: $\rho_1(\theta)=r_1^\theta$ for all $\theta\in\Theta$; $\rho_2(r_1,\theta)= r_2^-$ if $\theta<0$ and $\theta\notin r_1$, and $\rho_2(r_1,\theta)= r_2^\theta$ otherwise. Beliefs $p$ are such that $\beta(r_1,r_2)=a^+$ if $\theta\geq 0$ for all $\theta\in r_1$ and $C_2(r_2,\theta)<|\Delta u_2(\theta)|$ for some $\theta\geq 0$;  $\beta(r_1,r_2)=a^-$ otherwise. The proof of Proposition~\ref{prop:eqm} applies directly to this case as well.

For example, suppose the report space is $R=\{a_-^-,a^-,a^+\}$, with $C_j(a^+,\theta)=C_j(a^-,\theta')=0$ for all $\theta\geq 0$, $\theta'<0$, and $j\in\{1,2\}$. We can say that $\theta\in r_j$ if and only if $C_j(r_j,\theta)=0$. For sender~2, the report $r_2=a_-^-$ cannot be profitably delivered when the state is non-negative, as $C_2(a_-^-,\theta)>|\Delta u_2(\theta)|$ for all $\theta \geq 0$. However, $C_2(a_-^-,\theta)<|\Delta u_2(\theta)|$ for all $\theta < 0$. There is an efficient equilibrium where sender~1 reports $r_1=a^+$ when $\theta\geq 0$, and $r_1=a^-$ otherwise. Sender~2 delivers $r_2=a_-^-$ if and only if $\theta<0$ and $r_1 = a^+$. Otherwise, it delivers $r_2=a^+$ when $\theta\geq 0$ and $r_2=a^-$ when $\theta<0$. Beliefs are such that the receiver takes action $a^+$ if and only if $r_1=a^+$ and $C_2(r_2,\theta)<|\Delta u_2(\theta)|$ for some $\theta\geq 0$. Otherwise, the receiver selects $a^-$. These strategies and beliefs constitute an efficient equilibrium.

\subsection{More than two actions}\label{app:more_actions}

This section discusses the applicability of the paper's findings to settings with more than two actions. The binary-actions modeling choice introduces an element of limited liability: the receiver cannot discipline opposed-biased senders with the threat of taking actions that hurt them both when they deviate from truthful reporting. Punishing one sender would simultaneously reward the other one. Consider an extension of the baseline model where the receiver has an additional, ``safe'' alternative, $a^*$ that yields to each player a fix, state-independent payoff. The following example illustrates that the receiver may exploit the use of additional alternatives to discipline senders' behavior with the goal of sustaining truthful reporting, and thus inducing efficient outcomes.

Consider a \emph{simultaneous} communication protocol with two opposed-biased senders. Sender~$1$'s net payoff is such that $u_1(a^+,\theta)=1$ and $u_1(a^-,\theta)=0$ for all $\theta\in\Theta$. Sender~$2$'s net payoff is such that $u_2(a^+,\theta)=0$ and $u_2(a^-,\theta)=1$ for all $\theta\in\Theta$. The receiver obtains a payoff of~$1$ (resp.~$-1$) when selecting~$a^+$ (resp.~$a^-$) in non-negative states, or~$a^-$ (resp.~$a^+$) in negative states. In addition, the receiver can select a third alternative, $a^*$, which grants player~$i$ with a state-independent payoff of $v_i = u_i(a^*,\theta)$ for every $\theta\in\Theta$. To make the problem non-trivial, assume that $v_r<1$.

Suppose that this protocol supports an efficient equilibrium, where $r_1=r_2=\theta$ for every $\theta\in\Theta$. Off-path beliefs following $r_1\neq r_2$ are such that the state is non-negative with some probability $q(\mathbf{r})\in(0,1)$, where $\mathbf{r}=(r_1,r_2)$. When off-path, the receiver's expected utility  is $2q(\mathbf{r})-1$ from selecting $a^+$, and $1-2q(\mathbf{r})$ from choosing $a^-$. Consider the case where off-path beliefs are report-independent, that is, $q(\mathbf{r})=q$ for every $r_1\neq r_2$.

Let's introduce two further assumptions about preferences over the additional alternative: (i) $a^*$ is the receiver's preferred alternative when senders deliver different reports, i.e.,  $v_r> \max\{2q-1,1-2q\}$; (ii) $a^*$ is the senders' least-favorite option regardless of the state, that is, $v_i<0$ for $i\in\{1,2\}$. 

This extension supports an efficient equilibrium. Any deviation from truthful reporting results in the worst possible outcome for the senders. When reports do not coincide, the receiver selects $a^*$, and this threat is credible because choosing $a^*$ is sequentially rational in the postulated equilibrium of this example. Thus, a third alternative can enable efficient outcomes even under simultaneous protocols. Since the receiver can discipline senders to deliver matching reports but cannot verify their truthfulness, this mechanism is effective only with oppositely biased senders. By contrast, like-biased senders can coordinate persuasion by submitting matching reports, as shown in Section~\ref{sec:protocols}.

Intuitively, introducing a third alternative can undermine liability constraints inherent in the baseline binary-action framework. When that is the case, two oppositely biased senders are sufficient to achieve efficiency, regardless of whether communication is sequential or simultaneous. However, the paper's main result can still hold in settings with multiple actions, provided these additional alternatives meet certain conditions that preserve an element of limited liability. A straightforward case where public advocacy remains the only efficient minimal protocol is when the additional actions are never sequentially rational, both on and off the equilibrium path.


\addcontentsline{toc}{section}{References}
\bibliographystyle{apacite}
\bibliography{biblio_sequential}
\end{document}